\documentclass[oneside,english,reqno,12pt]{amsart}

\usepackage{geometry}
\usepackage{babel}
\usepackage{enumitem}
\usepackage{amssymb,amsmath,amsthm,amsfonts,latexsym}
\usepackage{bbm}
\usepackage{appendix}
\DeclareMathOperator{\tr}{Tr}
\usepackage{hyperref}
\usepackage{color}
\makeatletter
\numberwithin{equation}{section}
\numberwithin{figure}{section}
\theoremstyle{theorem}
\newtheorem{theorem}{\protect\theoremname}
\theoremstyle{corollary}

\theoremstyle{remark}
\newtheorem{remark}[theorem]{\protect\remarkname}
\theoremstyle{lemma}
\newtheorem{lemma}[theorem]{\protect\lemmaname}
\newlist{casenv}{enumerate}{4}
\setlist[casenv]{leftmargin=*,align=left,widest={iiii}}
\setlist[casenv,1]{label={{\itshape\ \casename} \arabic*.},ref=\arabic*}
\setlist[casenv,2]{label={{\itshape\ \casename} \roman*.},ref=\roman*}
\setlist[casenv,3]{label={{\itshape\ \casename\ \alph*.}},ref=\alph*}
\setlist[casenv,4]{label={{\itshape\ \casename} \arabic*.},ref=\arabic*}

\makeatother

\providecommand{\corollaryname}{Corollary}
\providecommand{\lemmaname}{Lemma}
\providecommand{\remarkname}{Remark}
\providecommand{\casename}{Case}
\providecommand{\theoremname}{Theorem}

\begin{document}
	
	\title[Blow-up Profile of Neutron Stars in the HFB theory]{Blow-up Profile of Neutron Stars \\ in the Hartree--Fock--Bogoliubov theory}
	
	\author[D.-T. Nguyen]{Dinh-Thi Nguyen}
\address{Dinh-Thi Nguyen, Mathematisches Institut, Ludwig--Maximilians--Universit\"at M\"unchen (LMU), Theresienstrasse 39, 80333 Munich, Germany, and Munich Center for Quantum Science and Technology (MCQST), Schellingstrasse 4, 80799 Munich, Germany.} 
\email{\href{mailto:nguyen@math.lmu.de}{nguyen@math.lmu.de}}
	
	
	\maketitle
	
	\begin{abstract}
		We consider the gravitational collapse for neutron stars in the Hartree--Fock--Bogoliubov theory. We prove that when the number particle becomes large and the gravitational constant is small such that the attractive interaction strength approaches the \emph{Chandrasekhar limit mass} slowly, the minimizers develop a universal blow-up profile. It is given by the Lane--Emden solution.
	\end{abstract}
	
	\section{Introduction}
	
	\label{intro}
	
	A neutron star is a relativistic system of identical fermions in $\mathbb{R}^3$ with Newtonian gravitational interaction. From the first principles of quantum mechanics, such a system is typically described by the $N$-particle Hamiltonian
	\begin{equation}\label{hamiltonian}
	H_{N} = \sum_{i=1}^{N}\sqrt{-\Delta_{x_i}+m^{2}}-\kappa\sum_{1\leq i<j \leq N}|x_i-x_j|^{-1}
	\end{equation}
	acting on $\bigwedge^N L^{2}(\mathbb{R}^3,\mathbb{C}^q)$, the Hilbert space of square-integrable functions which are anti-symmetric under the permutations of space-spin variables $q\geq 1$ ($q=2$ in nature). Here $m>0$ is the neutron mass (we choose the unit $\hbar=c=1$) and $\kappa=Gm^{2}$ with $G$ the gravitational constant. 
	
	It is a fundamental fact that the neutron star collapses (namely $H_{N}$ is not bounded from below) if the particle number is too big, such that
	$$
	\tau:=\kappa N^{2/3} > \tau_c.
	$$ 
	The critical constant $\tau_c$ was first computed by  Chandrasekhar \cite{Chandrasekhar-31a} using an effective semi-classical theory, and then confirmed rigorously by Lieb and Yau \cite{LieYau-87} using the many-body Schr\"odinger theory. In fact, $\tau_{c}$ is the optimal constant in the Hardy--Littlewood--Sobolev inequality
	\begin{equation}\label{ineq:HLS}
	\frac{\tau_{c}}{2}\iint_{\mathbb{R}^3\times \mathbb{R}^3}\frac{\rho(x)\rho(y)}{|x-y|}{\rm d}x{\rm d}y \leq K_{\rm cl}\|\rho\|_{L^{4/3}}^{4/3}\|\rho\|_{L^1}^{2/3},\quad \forall \rho\in L^{1}\cap L^{4/3}(\mathbb{R}^3),
	\end{equation}
	where $K_{\rm cl}=\frac{3}{4}(6\pi^{2}/q)^{1/3}$. Numerically, the proportion $\sigma_{f} := K_{\rm cl}\tau_{c}^{-1}$ is about $1.092$.
	
	It is well-known (see \cite[Appendix A]{LieOxf-80}) that \eqref{ineq:HLS} has a minimizer $Q\in L^{1}\cap L^{4/3}(\mathbb{R}^3)$ which is unique up to dilations and translations. Such $Q$ can be chosen uniquely to be non-negative symmetric decreasing by rearrangement inequality (see \cite[Chapter 3]{LieLos-01}) and it satisfies
	\begin{equation}\label{cond:LE}
	\sigma_{f}\int_{\mathbb{R}^3}Q(x)^{4/3}{\rm d}x = \int_{\mathbb{R}^3}Q(x){\rm d}x = \frac{1}{2}\iint_{\mathbb{R}^3\times \mathbb{R}^3}\frac{Q(x)Q(y)}{|x-y|}{\rm d}x{\rm d}y = 1.
	\end{equation}
	Moreover, $Q$ solves the Lane--Emden equation of order $3$,
	\begin{equation}\label{eq:massless-neutron-star}
	\frac{4}{3}\sigma_{f} Q(x)^{1/3}-(|\cdot|^{-1}\star Q)(x)+\frac{2}{3} \begin{cases}
	=0 & \text{if } Q(x)>0,\\
	\geq0 & \text{if } Q(x)=0.
	\end{cases}
	\end{equation}
	The Lane--Emden equation goes back to \cite{Lane-70} (see \cite{GadBarHan-80,Thomas-27} for detailed studies). Note that it can be easily seen from \eqref{eq:massless-neutron-star} that $Q$ has compact support (see \cite[Appendix A]{LieOxf-80}).
	
	The critical value $\tau_c$ can be obtained easily from the Chandrasekhar theory, a semi-classical approximation of the full many-body theory. In this effective theory, the ground state energy of a neutron star is given by
	\begin{equation}\label{eq:Ch-energy}
	E_{\tau}^{\rm Ch}(1) := \inf\left\{\mathcal{E}_{\tau}^{\rm Ch}(\rho) : 0\leq\rho\in L^{1}\cap L^{4/3}(\mathbb{R}^{3}),\int_{\mathbb{R}^3} \rho(x){\rm d}x=1\right\},
	\end{equation}
	where the Chandrasekhar functional is
	\begin{equation}\label{eq:Ch-functional}
	\mathcal{E}_{\tau}^{\rm Ch}(\rho) = \int_{\mathbb{R}^3}\frac{q}{(2\pi)^3}\int_{|p|<(6\pi^{2}\rho(x)/q)^{1/3}}\sqrt{|p|^{2}+m^{2}}{\rm d}p{\rm d}x - \frac{\tau}{2} \iint_{\mathbb{R}^3\times \mathbb{R}^3}\frac{\rho(x)\rho(y)}{|x-y|}{\rm d}x{\rm d}y.
	\end{equation}
	It can be easily seen from \eqref{ineq:HLS} that $E_{\tau}^{\rm Ch}(1)>-\infty$ if and only if $\tau\leq\tau_c$.
	
	In the seminal paper \cite{LieYau-87}, Lieb and Yau proved that for any fixed $\tau=\kappa N^{2/3} < \tau_c$, the quantum energy converges to the semi-classical energy
	$$
	\lim_{N\to\infty}\inf\text{spec}\frac{H_{N}}{N} = E_{\tau}^{\rm Ch}(1).
	$$
	See also \cite{LieThi-84} for an earlier related result and \cite{FouLewSol-18} for a recent extension to general interaction potentials. 
	
	In the present paper we are interested in the ground states of the neutron star in the critical regime, when $\tau=\kappa N^{2/3}\nearrow \tau_c$ simultaneously as $N\to \infty$. It turns out that the many-body theory is very complicated to study: the ground state does not exist due to the translation invariance, and even if we consider approximate ground states (in an appropriate sense), then their behavior is very unstable since the system can easily split into many small pieces without lowering much the energy. Therefore, it is reasonable to focus on some effective models where physical properties are easier to observe thanks to non-linear effects. 
	
	In the following, we will focus on the blow-up phenomenon of neutron stars in the Hartree--Fock--Bogoliubov (HFB) theory. This is one of the most important approximation methods in quantum mechanics, and it is a generalization of the traditional Hartree--Fock (HF) theory, taking into account all quasi-free states in Fock space. We refer to Bach, Lieb and Solovej \cite{BacLieSol-94}  for a general discussion on the derivation of the HFB theory  from many-body quantum mechanics (see also Bach, Fr\"{o}hlich and Jonsson in \cite{BacFroJon-09} for a simplification). In this model, we study the HFB energy functional given by
	\begin{equation}\label{HFB-functional}
	\mathcal{E}_{\tau}^{\rm HFB}(\gamma,\alpha) = \tr\sqrt{-\Delta+m^{2}}\gamma - \frac{\kappa}{2}  D(\rho_{\gamma},\rho_{\gamma}) + \frac{\kappa}{2} {\rm Ex}(\gamma) - \frac{\kappa}{2} \iint_{\mathbb{R}^3\times \mathbb{R}^3}\frac{|\alpha(x,y)|^{2}}{|x-y|}{\rm d}x{\rm d}y.
	\end{equation}
	Here we use the  subscript $\tau=\kappa N^{2/3}$ and the shorthand notations
	$$
	D(\rho_{1},\rho_{2}) = \iint_{\mathbb{R}^3\times \mathbb{R}^3}\frac{\rho_{1}(x)\rho_{2}(y)}{|x-y|}{\rm d}x{\rm d}y \quad \text{and} \quad {\rm Ex}(\gamma) = \iint_{\mathbb{R}^3\times \mathbb{R}^3}\frac{|\gamma(x,y)|^{2}}{|x-y|}{\rm d}x{\rm d}y,
	$$
	which we refer to as the direct term and the exchange term, respectively. The density matrix $\gamma$ is a self-adjoint, non-negative operator on $L^{2}(\mathbb{R}^3,\mathbb{C}^q)$ with $\tr\gamma=N$. The pairing density matrix $\alpha$ is a Hilbert--Schmidt operator on $L^{2}(\mathbb{R}^3,\mathbb{C}^q)$, i.e. $\tr\alpha^{*}\alpha < \infty$, and its kernel is a $(2\times 2)$-matrix which is supposed to be anti-symmetric in the sense $\alpha^T=-\alpha$. The set of HFB states is given by
	\begin{equation}\label{ineq:Gamma}
	\mathcal{K}_{\rm HFB} = \left\{(\gamma,\alpha)=(\gamma^{*},-\alpha^T)\in\mathcal{X}_{\rm HFB}:\left(\begin{array}{cc}0 & 0\\0 & 0\end{array}\right) \leq \left(\begin{array}{cc}\gamma & \alpha\\\alpha^{*} & 1-\overline{\gamma}\end{array}\right) \leq \left(\begin{array}{cc}1 & 0\\0 & 1\end{array}\right)\right\},
	\end{equation}
	where the Sobolev-type space $\mathcal{X}_{\rm HFB}$ is defined by
	$$
	\mathcal{X}_{\rm HFB} := \left\{(\gamma,\alpha)\in\mathfrak{S}_1\times\mathfrak{S}_2: \|(1-\Delta)^{1/4}\gamma (1-\Delta)^{1/4}\|_{\mathfrak{S}_1}+\|(1-\Delta)^{1/4}\alpha\|_{\mathfrak{S}_2}<\infty\right\}.
	$$
	The HFB minimization problem then reads
	\begin{equation}\label{eq:HFB-energy}
	E_{\tau}^{\rm HFB}(N) = \inf\big\{\mathcal{E}_{\tau}^{\rm HFB}(\gamma,\alpha):(\gamma,\alpha)\in\mathcal{K}_{\rm HFB},\tr\gamma=N\big\}.
	\end{equation}
	
	In the stable regime, the existence of a minimizer for the variational problem \eqref{eq:HFB-energy} has been proved by Lenzmann and Lewin \cite{LenLew-10}. The HFB energy $E_{\tau}^{\rm HFB}(N)$ is attained for $0<N<N^{\rm HFB}(\kappa)$,  $0<\kappa<\pi/4$ and $m>0$. The finite number $N^{\rm HFB}(\kappa)$ is provided by the asymptotic estimate $N^{\rm HFB}(\kappa) \sim (\tau_{c}/\kappa)^{3/2}$ as $\kappa \to 0$. The authors in \cite{LenLew-10} also proved that, for every $0<\tau<\tau_c$, we have
	\begin{equation}
	\lim_{\substack{N\to\infty \\ \kappa N^{2/3} = \tau}}\frac{E_{\tau}^{\rm HFB}(N)}{N} =  E_{\tau}^{\rm Ch}(1).
	\end{equation}
	Thus the HFB theory captures correctly the leading order of the many-body theory. Actually, this theory is believed to be a much better approximation to the full many-body Schr\"odinger theory than the Chandrasekhar theory.
	
	In this paper, we will focus on the case when $N\to \infty$ and $\tau:=\tau_{N}\nearrow \tau_c$ slowly and show that the HFB minimizers develop a universal blow-up profile given by the Lane--Emden solution. Our main result is
	\begin{theorem}[Blow-Up of HFB Minimizers]\label{thm:behavior}
		Let $q\geq 1$ be given and suppose that $m>0$. Assume that $0 < \tau_{N} = \tau_{c} - N^{-\beta}$ with $0<\beta<1/9$. Then
		\begin{equation} \label{collapse-HFB-energy}
		\frac{E_{\tau_{N}}^{\rm HFB}(N)}{N} =  (\tau_{c}-\tau_{N})^{1/2}(2\Lambda+o(1)_{N\to\infty}),
		\end{equation}
		where
		\begin{equation} \label{Lambda}
		\Lambda = \frac{3}{4}m\sqrt{\frac{1}{K_{\rm cl}}\int_{\mathbb{R}^3}Q(x)^{2/3}{\rm d}x}
		\end{equation}
		and $Q$ is the unique non-negative radial function satisfying \eqref{cond:LE}--\eqref{eq:massless-neutron-star}. Furthermore, if $(\gamma_{N},\alpha_{N})$ is a minimizer of $E_{\tau_{N}}^{\rm HFB}(N)$ and $\rho_{\gamma_{N}}(x)=\gamma_{N}(x,x)$, then there exist a subsequence of $\{\tau_{N}\}$, still denoted by $\{\tau_{N}\}$, and a sequence $\{y_{N}\}\subset\mathbb{R}^3$ such that
		\begin{equation} \label{collapse-HFB-ground-states}
		\lim_{N\to\infty}(\tau_{c}-\tau_{N})^{3/2}\rho_{\gamma_{N}}((\tau_{c}-\tau_{N})^{1/2}N^{1/3}x+y_{N}) = \Lambda^3Q(\Lambda x)
		\end{equation}
		strongly in $L^{r}(\mathbb{R}^3)$ for $1\leq r < 4/3$ and weakly in $L^{4/3}(\mathbb{R}^3)$.
		
	\end{theorem}
	
	\begin{remark}
		\begin{itemize}
			\item Attaining the $L^{4/3}(\mathbb{R}^3)$ convergence in \eqref{collapse-HFB-ground-states} would require, at least with our method, to prove the Lieb--Thirring inequality with the sharp constant as conjectured in \cite{LieYau-87} (see \cite{Daubechies-83} and \cite[Chapter 4]{LieSei-10} for thorough discussions).

			\item The contribution of the pairing term in \eqref{HFB-functional} is coupled by the small parameter $\kappa=\mathcal{O}(N^{-2/3})$. Therefore it does not show up in the leading order of the blow-up profile.

			\item Since the limit in \eqref{collapse-HFB-ground-states} is unique, we expect that the density of the HFB minimizer is unique, at least when $N$ is sufficient large. Probably this can be proved using techniques in the ground state problem of non-linear Schr\"odinger functionals (see e.g. \cite{AscFroGraSchTro-02,Maeda-10,FraLen-13,FraLenSil-16}), but it seems non-trivial. We hope to come back this issue in the future.
			\end{itemize}
\end{remark}

	The above result is a continuation of our work in \cite{Nguyen-19f} on the blow-up profile of neutron stars in the Chandrasekhar theory \eqref{eq:Ch-energy}, which is purely semi-classical (see also \cite{Nguyen-19b,Nguyen-17a,Nguyen-17b,GuoZen-17,YanYan-17} for discussions on the blow-up profile of boson stars). The HFB theory is believed to be much more precise than the Chandrasekhar theory, and the analysis in this case is also significantly more difficult. The proof of \cite[Theorem 2]{Nguyen-19f} is based on a detailed analysis of the Euler--Lagrange equation associated to the minimizers for $E_{\tau}^{\rm Ch}(1)$ when $\tau \nearrow \tau_{c}$. In this paper, our proof of Theorem \ref{thm:behavior} is based on the concentration compactness method \cite{Lions-84a,Lions-84b}. In contrast to the classical dichotomy argument, the relative compactness of the sequence of the densities of the HFB minimizers is not a consequence of the strict binding inequality, but it comes from the non-existence of minimizers in the variational problem $E^{\rm Ch}_{\tau_c}(\nu)|_{m=0}=0$ with $0<\nu<1$. By the same method, we can also extend the blow-up result in the Chandrasekhar theory to the general approximate Chandrasekhar minimizer. Finally, we remark that the blow-up phenomenon of neutron stars has also been studied in the time-dependent setting (see \cite{FroLen-07f,HaiSch-09,Hainzl-10,HaiLenLewSch-10} for rigorous results). This problem, however, is different from the ground state problem that we consider in the present paper.
	
\medskip

\textbf{Organization of the paper.} In Section \ref{sec:energy estimate} we establish some estimates for the HFB energy via the Chandrasekhar energy, and a moment estimate for the densities of the HFB minimizers. In Section \ref{sec:behavior}, we prove Theorem \ref{thm:behavior} which gives the blow-up profile of minimizers for the HFB minimization problem \eqref{eq:HFB-energy}.

	\section{Energy Estimates} \label{sec:energy estimate}

	Since the full many-body Schr\"odinger theory of neutron stars is very complicated, approximate but simpler theories are often used to study the stellar collapse for neutron stars. The simplest approximate theory is the semi-classical Chandrasekhar theory \eqref{eq:Ch-energy}. This has been rigorously derived from many-body quantum mechanics by Lieb and Yau in \cite{LieYau-87} (see also \cite{LieThi-84}). In this section, we revisit the blow-up phenomenon in the Chandrasekhar theory. Note that the kinetic energy functional in \eqref{eq:Ch-functional} can be calculated explicitly as follows
	\begin{align}
	j_{m}(\rho) :& = \frac{q}{(2\pi)^3}\int_{|p|<(6\pi^{2}\rho/q)^{1/3}}\sqrt{|p|^{2}+m^{2}}{\rm d}p \label{kinetic}\\
	& = \frac{q}{16\pi^{2}}\left[\eta(2\eta^{2}+m^{2})\sqrt{\eta^{2}+m^{2}}-m^4\ln\left(\frac{\eta+\sqrt{\eta^{2}+m^{2}}}{m}\right)\right],\quad \eta=\left(\frac{6\pi^{2}\rho}{q}\right)^{1/3}. \nonumber
	\end{align}

	For the reader's convenience, we recall the following results on the existence and uniqueness of the Chandrasekhar minimizer (see \cite[Theorem 3]{LieYau-87}) and the blow-up profile of neutron stars in the Chandrasekhar theory (see \cite[Theorem 2]{Nguyen-19f}).

\begin{theorem}[Existence of the Chandrasekhar Minimizer] 
	\label{thm:existence-Ch-minimizer} 	Let $q\geq 1$ be given and suppose that $m>0$. Then the variational problem $E^{\rm Ch}_{\tau}(1)$ in \eqref{eq:Ch-energy} has the following properties.
	\begin{itemize}
		\item[(i)] If $0< \tau<\tau_{c}$ then $E^{\rm Ch}_{\tau}(1) \geq 0$ and it has a unique minimizer (up to translation). The minimizer can be chosen to be radially symmetric decreasing,
		\item[(ii)] If $\tau=\tau_{c}$ then $E^{\rm Ch}_{\tau}(1)=0$ but it has no minimizer,
		\item[(iii)] If $\tau>\tau_{c}$ then $E^{\rm Ch}_{\tau}(1)=-\infty$.
	\end{itemize}
\end{theorem}

\begin{remark}
	The Chandrasekhar minimizer $\rho$ satisfies the Euler--Lagrange equation, for some Lagrange multiplier $\mu$,
	$$
	j_{m}'(\rho(x)) = (\eta(x)^{2}+m^{2})^{1/2} = [\tau|x|^{-1}\star\rho-\mu]_{+}  ,
	$$
	where $[f(x)]_{+}=\max\{f(x),0\}$ and $\eta=(6\pi^{2}\rho/q)^{1/3}$. This equation is equivalent to	the Newtonian limit of the Tolman--Oppenheimer--Volkoff equation \cite{Weinberg-72,Straumann-12}.
\end{remark}

\begin{theorem}[Blow-Up of the Chandrasekhar Minimizer] \label{thm:behavior-Ch-minimier} 
	Let $q\geq 1$ be given and suppose that $m>0$. Let $\rho_{\tau}$ be the unique minimizer (up to translation) of $E^{\rm Ch}_{\tau}(1)$ for $0< \tau<\tau_{c}$. Then for every sequence $\{\tau_{n}\}$ with $\tau_{n}\nearrow \tau_{c}$ as $n\to\infty$, we have
	\begin{equation}\label{asymptotic:Ch-energy}
	E^{\rm Ch}_{\tau_n}(1) =  (\tau_{c}-\tau_{n})^{1/2}(2\Lambda+o(1)_{n\to\infty}),
	\end{equation}
	where $\Lambda$ is determined as in \eqref{Lambda}. Furthermore, there exist a subsequence of $\{\tau_{n}\}$, still denoted by $\{\tau_{n}\}$, and a sequence $\{y_{n}\}\subset \mathbb{R}^3$ such that
	\[
	\lim_{n\to\infty}(\tau_{c}-\tau_{n})^{3/2}\rho_{\tau_{n}}((\tau_{c}-\tau_{n})^{1/2}x+y_{n})=\Lambda^{3}Q\left(\Lambda x\right)
	\]
	strongly in $L^{1}\cap L^{4/3}(\mathbb{R}^3)$. Here $Q$ is the unique non-negative radial function satisfying \eqref{cond:LE}--\eqref{eq:massless-neutron-star}.
\end{theorem}

	In \cite{Nguyen-19f}, the proof of Theorem \ref{thm:behavior-Ch-minimier}  is based on a detailed analysis of the Euler--Lagrange equation for the exact Chandrasekhar minimizer. This can be extended to the general approximate Chandrasekhar minimizer. The following result can be proved by adapting our arguments in the next section and the arguments in \cite{Nguyen-19f}.

\begin{theorem}[Blow-Up of the Approximate Chandrasekhar Minimizers] \label{thm:behavior-Ch-approximate}
	Let $q\geq 1$ be given and suppose that $m>0$. Let $\tau_{n} \nearrow \tau_c$ as $n\to\infty$ and let $\rho_{n} \in L^{1}\cap L^{4/3}(\mathbb{R}^3)$ be a sequence of non-negative functions such that $\int_{\mathbb{R}^3}\rho_{n}(x){\rm d}x = 1$ and 
	\begin{equation}\label{eq:Ch-energy-approximate-1}
	\lim_{n\to\infty}\frac{\mathcal{E}_{\tau_{n}}^{\rm Ch}(\rho_{n})}{E_{\tau_{n}}^{\rm Ch}(1)} = 1.
	\end{equation}
	Then there exist a subsequence of $\{\tau_{n}\}$, still denoted by $\{\tau_{n}\}$, and a sequence $\{y_{n}\}\subset \mathbb{R}^3$ such that  
	$$
	\lim_{n\to\infty}(\tau_{c}-\tau_{n})^{3/2}\rho_{\tau_{n}}((\tau_{c}-\tau_{n})^{1/2}x+y_{n}) = \Lambda^{3}Q(\Lambda x)
	$$
	strongly in $L^{1}\cap L^{4/3}(\mathbb{R}^{3})$. Here $\Lambda$ is determined as in \eqref{Lambda} and $Q$ is the unique non-negative radial function satisfying \eqref{cond:LE}--\eqref{eq:massless-neutron-star}.
\end{theorem}

	The aim of this section is to show that the HFB energy $N^{-1}E_{\tau}^{\rm HFB}(N)$ has the same asymptotic behavior as the Chandrasekhar energy in \eqref{asymptotic:Ch-energy} when $N\to\infty$ and $\tau:=\tau_{N}=\kappa N^{2/3} \nearrow \tau_c$ slowly. We note that the operator inequality of HFB states in \eqref{ineq:Gamma} holds on $L^{2}(\mathbb{R}^3,\mathbb{C}^q)\oplus L^{2}(\mathbb{R}^3,\mathbb{C}^q)$. This guarantees that the pair $(\gamma,\alpha)$ is associated to a unique quasi-free state in Fock space (see \cite{BacLieSol-94}). Also, it leads to the operator inequality (see \cite{BacLieSol-94})
\begin{equation}\label{ineq:gamma-alpha}
\gamma^{2} +\alpha\alpha^{*} \leq \gamma.
\end{equation}
The basic fact $0\leq \gamma\leq 1$ refers to the Pauli exclusion principle \cite[Theorem 3.2]{LieSei-10} (see also \cite{Lieb-83}).

\begin{lemma}[Collapse of the HFB Energy]\label{lem:HFB-energy}
	Let $q\geq 1$ be given and suppose that $m>0$. Let $0 < \tau_{N} = \tau_{c} - N^{-\beta}$ with $0<\beta<1/9$. Then we have
\begin{equation}\label{ineq:HFB-energy}
\frac{E_{\tau_{N}}^{\rm HFB}(N)}{N} = E_{\tau_{N}}^{\rm Ch}(1) + o(E_{\tau_{N}}^{\rm Ch}(1))_{N\to\infty} = (\tau_{c}-\tau_{N})^{1/2} \left(2\Lambda + o(1)_{N\to\infty}\right).
\end{equation}
\end{lemma}

\begin{proof}
We start with the lower bound. Let $(\gamma_{N},\alpha_{N})$ be a minimizer of $E_{\tau_{N}}^{\rm HFB}(N)$ for $0<\tau_{N}=\kappa N^{2/3} < \tau_{c}$. Applying the Hardy--Kato inequality $|x|^{-1}\leq \frac{\pi}{2}\sqrt{-\Delta}$ (see \cite{Kato-95,Herbst-77}) in the variable $x$ with $y$ fixed and using \eqref{ineq:gamma-alpha} we obtain
\begin{equation}\label{ineq:pairing-exchange}
\iint_{\mathbb{R}^3\times \mathbb{R}^3}\frac{|\alpha_{N}(x,y)|^{2}}{|x-y|}{\rm d}x{\rm d}y \leq \frac{\pi}{2}\tr\sqrt{-\Delta}\alpha_{N}\alpha_{N}^{*} \leq \frac{\pi}{2}\tr\sqrt{-\Delta}\gamma_{N}.
\end{equation}
It follows from \eqref{ineq:pairing-exchange} and the non-negativity of the exchange term that
\begin{equation}\label{energy:lower}
E_{\tau_{N}}^{\rm HFB}(N) = \mathcal{E}_{\tau_{N}}^{\rm HFB}(\tau_{N},\alpha_{N}) \geq \left(1-\frac{\kappa\pi}{4}\right)\tr\sqrt{-\Delta+m^{2}}\gamma_{N} - \frac{\kappa}{2} D(\rho_{\gamma_{N}},\rho_{\gamma_{N}}).
\end{equation}
By the arguments in \cite[Proof of Theorem 1]{LieYau-87} we have
\begin{equation}\label{ineq:lower-bound-1}
\frac{E_{\tau_{N}}^{\rm HFB}(N)}{N} \geq E_{\tau_{N}'}^{\rm Ch}(1) - 2\epsilon m,
\end{equation}
where $\tau_{N}' = \kappa'N^{2/3} < \tau_c$ with $\kappa' = \kappa(1-\kappa\pi/4-\epsilon)^{-1}$ and $\epsilon = 1.7q^{1/3}\kappa^{2/3}N^{1/3}=\mathcal{O}(N^{-1/9})$. We deduce from the asymptotic formula for $E_{\tau_{N}}^{\rm Ch}(1)$ in \eqref{asymptotic:Ch-energy} that
\begin{align}\label{ineq:lower-bound-2}
E_{\tau_{N}'}^{\rm Ch}(1) - E_{\tau_{N}}^{\rm Ch}(1) & = ((\tau_{c}-\tau_{N}')^{1/2} - (\tau_{c}-\tau_{N})^{1/2}) \left(2\Lambda + o(1)_{N\to\infty}\right) \nonumber \\
& \geq -(\tau_{N} '-\tau_{N})^{1/2} \left(2\Lambda + o(1)_{N\to\infty}\right).
\end{align}
Since $\kappa=\tau_{N}N^{-2/3} = \mathcal{O}(N^{-2/3})$ we have $\tau_{N} '-\tau_{N} = \mathcal{O}(N^{-1/9})$. Thus, it follows from \eqref{ineq:lower-bound-1}, \eqref{ineq:lower-bound-2} and \eqref{asymptotic:Ch-energy} that
$$
\frac{E_{\tau_{N}}^{\rm HFB}(N)}{N} \geq E_{\tau_{N}}^{\rm Ch}(1)(1 - \mathcal{O}(N^{-1/18})(\tau_{c}-\tau_{N})^{-1/2}).
$$
The error term $\mathcal{O}(N^{-1/18})(\tau_{c}-\tau_{N})^{-1/2}$ is of order $1$ when $\tau_{c} - \tau_{N} = N^{-\beta}$ with $0<\beta<1/9$.

Now we turn to the upper bound. Again, applying the Hardy--Kato inequality (see \cite{Kato-95,Herbst-77}) in the variable $x$ with $y$ fixed and using \eqref{ineq:gamma-alpha} we have
\begin{equation}\label{estimate:upper-energy-1}
{\rm Ex}(\gamma) = \iint_{\mathbb{R}^3\times \mathbb{R}^3}\frac{|\gamma(x,y)|^{2}}{|x-y|}{\rm d}x{\rm d}y \leq \frac{\pi}{2}\tr\sqrt{-\Delta}\gamma^{2} \leq \frac{\pi}{2}\tr\sqrt{-\Delta}\gamma.
\end{equation}
On the other hand, for any $0<\tau_{N}<\tau_{c}$, let $\rho_{N}^{\rm Ch}$ be the unique minimizer (up to translation) of $E_{\tau_{N}}^{\rm Ch}(1)$ in \eqref{eq:Ch-energy}. Since $D(\rho_{\gamma}-\rho_{N}^{\rm Ch}(N^{-1/3}\cdot),\rho_{\gamma}-\rho_{N}^{\rm Ch}(N^{-1/3}\cdot))\geq 0$, one have
\begin{equation}\label{estimate:upper-energy-2}
-D(\rho_{\gamma},\rho_{\gamma}) \leq D(\rho_{N}^{\rm Ch}(N^{-1/3}\cdot),\rho_{N}^{\rm Ch}(N^{-1/3}\cdot)) - 2D(\rho_{N}^{\rm Ch}(N^{-1/3}\cdot),\rho_{\gamma}).
\end{equation}
We deduce from \eqref{estimate:upper-energy-1}, \eqref{estimate:upper-energy-2} and the non-negativity of the pairing term that
\begin{align}
E_{\tau_{N}}^{\rm HFB}(N) \leq & \inf_{(\gamma,0)\in\mathcal{K}_{\rm HFB},\tr\gamma=N}\left\{\tr\left[\left(1+\frac{\kappa\pi}{4}\right)\sqrt{-\Delta+m^{2}} - \kappa\rho_{N}^{\rm Ch}(N^{-1/3}\cdot)\star |\cdot|^{-1}\right]\gamma\right\} \nonumber\\
& + \frac{\kappa}{2} D(\rho_{N}^{\rm Ch}(N^{-1/3}\cdot),\rho_{N}^{\rm Ch}(N^{-1/3}\cdot)). \label{ener:upper}
\end{align}
By choosing the trial state $\gamma := \mathbbm{1}(\sqrt{-\Delta}\leq \eta_{N}^{\rm Ch})
$ with $\eta_{N}^{\rm Ch}(x) = (6\pi^{2}\rho_{N}^{\rm Ch}(N^{-1/3}x)/q)^{1/3}$, we obtain
\begin{align}\label{ener:estimate-upper}
E_{\tau_{N}}^{\rm HFB}(N) \leq & \int_{\mathbb{R}^3}\frac{q}{(2\pi)^3}\int_{|p|\leq \eta_{N}^{\rm Ch}(x)}\left(1+\frac{\kappa\pi}{4}\right)\sqrt{|p|^{2}+m^{2}} - \kappa(\rho_{N}^{\rm Ch}(N^{-1/3}\cdot)\star |\cdot|^{-1})(x){\rm d}p{\rm d}x \nonumber \\
&  +  \frac{\kappa}{2} D(\rho_{N}^{\rm Ch}(N^{-1/3}\cdot),\rho_{N}^{\rm Ch}(N^{-1/3}\cdot)) \nonumber \\
= & \left(1+\frac{\kappa\pi}{4}\right)\int_{\mathbb{R}^3}j_{m}(\rho_{N}^{\rm Ch}(N^{-1/3}x)){\rm d}x - \frac{\kappa}{2} D(\rho_{N}^{\rm Ch}(N^{-1/3}\cdot),\rho_{N}^{\rm Ch}(N^{-1/3}\cdot)) \nonumber \\
= & N\left(E^{\rm Ch}_{\tau_{N}}(1) + \frac{\kappa\pi}{4}\int_{\mathbb{R}^3}j_{m}(\rho_{N}^{\rm Ch}(x)){\rm d}x\right).
\end{align}
We deduce from \eqref{ener:estimate-upper}, \eqref{ineq:HLS} and the asymptotic formula of $E_{\tau_{N}}^{\rm Ch}(1)$ in \eqref{asymptotic:Ch-energy} that
\begin{equation}\label{lem:moment}
\frac{E_{\tau_{N}}^{\rm HFB}(N)}{N} \leq E_{\tau_{N}}^{\rm Ch}(1) + \mathcal{O}(N^{-2/3})(\tau_{c}-\tau_{N})^{-1/2} = E_{\tau_{N}}^{\rm Ch}(1)(1 + \mathcal{O}(N^{-2/3})(\tau_{c}-\tau_{N})^{-1}).
\end{equation}
The error term $\mathcal{O}(N^{-2/3})(\tau_{c}-\tau_{N})^{-1}$ is of order $1$ when $\tau_{c} - \tau_{N} = N^{-\beta}$ with $0<\beta<2/3$.
\end{proof}

We now collect a moment estimate for the densities of the HFB minimizers by using \eqref{lem:moment}. This estimate will be useful for the proof of Theorem \ref{thm:behavior} in the next section.

\begin{lemma}[Moment Estimate] \label{lem:coercivity}
	Let $q\geq 1$ be given and suppose that $m>0$. Assume that $0 < \tau_{N} = \tau_{c} - N^{-\beta}$ with $0<\beta<1/9$. Let $(\gamma_{N},\alpha_{N})$ be a minimizer of $E_{\tau_{N}}^{\rm HFB}(N)$. Then we have
	\begin{equation}\label{ineq:coercivity}
	\frac{1}{N}\tr\sqrt{-\Delta+m^{2}}\gamma_{N} \leq C(\tau_{c}-\tau_{N})^{-1/2}.
	\end{equation}
\end{lemma}

\begin{proof}
	Let $\tau_{N}' = \tau_{N} + \mathcal{O}(N^{-1/9}) < \tau_{c}$. For any $0<\epsilon<1$ we have 
\begin{align}
E_{\tau_{N}}^{\rm Ch}(1) + \mathcal{O}(N^{-2/3}) & \geq \frac{E_{\tau_{N}}^{\rm HFB}(N)}{N} \nonumber \\
& \geq \epsilon \frac{1}{N}\tr\sqrt{-\Delta+m^{2}}\gamma_{N} + (1-\epsilon)\frac{1}{N}E_{\frac{\tau_{N}}{1-\epsilon}}^{\rm HFB}(N) \nonumber \\
& \geq \epsilon \frac{1}{N}\tr\sqrt{-\Delta+m^{2}}\gamma_{N} + (1-\epsilon)E_{\frac{\tau_{N}'}{1-\epsilon}}^{\rm Ch}(1) - \mathcal{O}(N^{-1/9}). \label{ineq:moment}
\end{align}
Here we have used \eqref{lem:moment} for the first estimate and \eqref{ineq:lower-bound-1} for the last estimate. Since $\tau_{N} = \tau_{c} - N^{-\beta}$ with $0<\beta<1/9$, we can choose $\epsilon$ small such that $(1-\epsilon)^{-1}\tau_{N}'<\tau_{c}$. Indeed, we can choose $2\epsilon = \tau_{c}^{-1}(\tau_{c}-\tau_{N})$.
This implies that $E_{\frac{\tau_{N}'}{1-\epsilon}}^{\rm Ch}(1)\geq 0$. Thus, \eqref{ineq:coercivity} is obtained from \eqref{ineq:moment} and the asymptotic formula for $E_{\tau_{N}}^{\rm Ch}(1)$ in \eqref{asymptotic:Ch-energy}.
\end{proof}

\begin{remark}\label{rem:coercivity}
	It follows from \eqref{ineq:coercivity} and Daubechies' inequality \cite{Daubechies-83} that
	$$
	\int_{\mathbb{R}^3}\rho_{\gamma_{N}}(N^{1/3}x)^{4/3}{\rm d}x = \frac{1}{N}\int_{\mathbb{R}^3}\rho_{\gamma_{N}}(x)^{4/3}{\rm d}x \leq C(\tau_{c}-\tau_{N})^{-1/2}.
	$$
\end{remark}

\section{Blow-Up of the HFB Minimizers} \label{sec:behavior}

In this section, we prove the convergence of the sequence of the densities of the HFB minimizers in Theorem \ref{thm:behavior}. We will need the following two lemmas.

\begin{lemma}\label{Chandrasekhar-lower-bound}
	Let $q\geq 1$ be given and suppose that $m\geq 0$. Let $g\in H^{1/2}(\mathbb{R}^3)$ with $\|g\|_{L^{2}}=1$. Then for any positive semi-definite operator $0\leq \gamma \leq 1$ and $\rho_{\gamma}(x)=\gamma(x,x)$ we have
	$$
	\tr\sqrt{-\Delta+m^{2}}\gamma \geq \int_{\mathbb{R}^3}j_{m}((\rho_{\gamma}\star |g|^{2})(x)){\rm d}x - \langle g,\sqrt{-\Delta}g \rangle \tr\gamma .
	$$
\end{lemma}
\begin{proof}
	The proof of this lemma can be found in \cite[Proof of Lemma B.3]{LieYau-87}.
\end{proof}

\begin{lemma}\label{lem:LT}
	Let $q\geq 1$ be given and let $0\leq \gamma_{N}\leq 1$ be a sequence of density matrix as a trace class operator such that the sequence $\rho_{\gamma_{N}}(N^{1/3}\cdot) = \gamma_{N}(N^{1/3}\cdot,N^{1/3}\cdot)$ converges to $\rho$ weakly in $L^{4/3}(\mathbb{R}^3)$. Then we have
	\begin{equation}\label{conv-ground-state-0}
	\liminf_{N\to\infty}\frac{1}{N}\tr\sqrt{-\Delta}\gamma_{N} \geq K_{\rm cl}\int_{\mathbb{R}^3}\rho(x)^{4/3}{\rm d}x.
	\end{equation}
\end{lemma}

\begin{proof}
	Let $\tilde{\gamma}_{N}(x,y) = N\gamma_{N}(N^{1/3}x,N^{1/3}y)$. For every function $0\leq V\in L^4(\mathbb{R}^3)$ we write
	\begin{equation}\label{conv-ground-state-1}
	\tr\sqrt{-\Delta}\gamma_{N} = N^{-1/3}\tr\sqrt{-\Delta}\tilde{\gamma}_{N} = N^{-1/3}\tr(\sqrt{-\Delta}-N^{1/3}V)\tilde{\gamma}_{N} + \tr V\tilde{\gamma}_{N}.
	\end{equation}
	By the assumption we have $0\leq \gamma_{N}\leq 1$ and hence $0\leq \tilde{\gamma}_{N}\leq 1$. We may apply the min-max principle \cite[Theorem 12.1]{LieLos-01} and Weyl's law on the sum of negative eigenvalues of Schr\"odinger operators (see \cite[Chapter 4]{LieSei-10}) to get the following estimate 
	\begin{align}\label{conv-ground-state-2}
	\tr(\sqrt{-\Delta}-N^{1/3}V)\tilde{\gamma}_{N}& \geq \tr[\sqrt{-\Delta}-N^{1/3}V]_{-} \nonumber\\
	& =  -\frac{N^{4/3}q|B_{\mathbb{R}^3}(0,1)|}{12(2\pi)^3}\left(\int_{\mathbb{R}^3}V(x)^{4}{\rm d}x+o(1)_{N\to\infty}\right).
	\end{align}
	On the other hand, it follows from the weak convergence $\rho_{\gamma_{N}}(N^{1/3}\cdot)\rightharpoonup\rho$ in $L^{4/3}(\mathbb{R}^3)$ that
	\begin{equation}\label{conv-ground-state-3}
	\lim_{N\to\infty}\frac{1}{N}\tr V\tilde{\gamma}_{N} = \lim_{N\to\infty}\int_{\mathbb{R}^3}V(x)\rho_{\gamma_{N}}(N^{1/3}x){\rm d}x = \int_{\mathbb{R}^3}V(x)\rho(x){\rm d}x.
	\end{equation}
	We deduce from \eqref{conv-ground-state-1}, \eqref{conv-ground-state-2} and \eqref{conv-ground-state-3} that
	$$
	\liminf_{N\to\infty}\frac{1}{N}\tr\sqrt{-\Delta}\gamma_{N} \geq -\frac{q|B_{\mathbb{R}^3}(0,1)|}{12(2\pi)^3}\int_{\mathbb{R}^3}V(x)^{4}{\rm d}x + \int_{\mathbb{R}^3}V(x)\rho(x){\rm d}x.
	$$
	Optimizing over $V$ leads to the desired lower bound \eqref{conv-ground-state-0}.
\end{proof}

From now on, we will denote $\ell_{N}=\Lambda(\tau_{c}-\tau_{N})^{-1/2}$. Let $\tilde{\gamma}_{N}(x,y) = \ell_{N}^{-3}\gamma_{N}(\ell_{N}^{-1}x,\ell_{N}^{-1}y)$ and $\rho_{\tilde{\gamma}_{N}}(x) = \tilde{\gamma}_{N}(x,x)$. Setting $w_{N}(x):=\rho_{\tilde{\gamma}_{N}}(N^{1/3}x)$ then $w_{N}$ is bounded uniformly in $L^{4/3}(\mathbb{R}^3)$, by Remark \ref{rem:coercivity}. The proof of \eqref{collapse-HFB-ground-states} in Theorem \ref{thm:behavior} is divided into several steps as follows.

\emph{Step 1: No vanishing.} We first rule out the vanishing of the sequence $\{w_{N}\}$. By vanishing we mean that
$$
\limsup_{N\to\infty}\Big(\sup_{y\in\mathbb{R}^3}\int_{|x-y|\leq R}w_{N}(x){\rm d}x\Big) = 0
$$
for all $R>0$. By the arguments in \cite[p.124]{Lions-84a} (see also \cite{LenLew-10,LenLew-11}) we obtain 
\begin{equation}\label{non-vanishing}
\lim_{N\to\infty}D(w_{N},w_{N}) = 0.
\end{equation}
On the other hand, for any $\tau_{N}'$ such that $0<\tau_{N}'<\tau_{N}<\tau_{c}$, we have
\begin{equation}\label{ineq:direct-term-1}
E_{\tau_{N}'}^{\rm HFB}(N) \leq \mathcal{E}_{\tau_{N}'}^{\rm HFB}(w_{N}) = E_{\tau_{N}}^{\rm HFB}(N) + \frac{\tau_{N}-\tau_{N}'}{2N^{2/3}}(D(\rho_{\gamma_{N}},\rho_{\gamma_{N}}) + CN\ell_{N}).
\end{equation}
Here we have used the non-negativity of the exchange term and the estimates using \eqref{ineq:pairing-exchange}, \eqref{ineq:coercivity} for the pairing term. Now we recall the following energy estimate in \cite[Lemma 7]{Nguyen-19f}
\begin{equation}\label{ineq:direct-term-2}
M_{1}(\tau_{c}-\tau_{N})^{1/2} \leq E_{\tau_{N}}^{\rm Ch}(1) \leq M_{2}(\tau_{c}-\tau_{N})^{1/2},
\end{equation}
where $M_{1}$ and $M_{2}$ are positive constants such that $M_{1}<M_{2}$. We deduce from \eqref{ineq:direct-term-1} and \eqref{ineq:direct-term-2} that
\begin{align*}
\frac{1}{2}D(w_{N},w_{N}) + \mathcal{O}(N^{-2/3}) & \geq \frac{E_{\tau_{N}'}^{\rm HFB}(N)-E_{\tau_{N}}^{\rm HFB}(N)}{N\ell_{N}(\tau_{N}-\tau_{N}')} \\
& \geq \frac{E_{\tau_{N}'+\mathcal{O}(N^{-1/9})}^{\rm Ch}(1)-E_{\tau_{N}}^{\rm Ch}(1)-\mathcal{O}(N^{-2/3})-\mathcal{O}(N^{-1/9})}{\ell_{N}(\tau_{N}-\tau_{N}')} \\
& \geq \frac{M_1(\tau_{c}-\tau_{N}'-\mathcal{O}(N^{-1/9}))^{1/2} - M_2(\tau_{c}-\tau_{N})^{1/2}-\mathcal{O}(N^{-1/9})}{\ell_{N}(\tau_{N}-\tau_{N}')} \\
& \geq \frac{M_1\left(\tau_{c}-\tau_{N}'\right)^{1/2} - M_2\left(\tau_{c}-\tau_{N}\right)^{1/2}-\mathcal{O}(N^{-1/18})}{\ell_{N}(\tau_{N}-\tau_{N}')}.
\end{align*}
Here we have used the estimates for the HFB energy as in the proof of Lemma \ref{lem:HFB-energy}. Choosing $\tau_{N}'=\tau_{N}-\delta(\tau_{c}-\tau_{N})$ with $\delta>0$ and recalling $\ell_{N}=\Lambda(\tau_{c}-\tau_{N})^{-1/2}$, we arrive at 
\begin{equation}\label{strict-non-vanishing}
\frac{1}{2}D(w_{N},w_{N}) + \mathcal{O}(N^{-2/3}) \geq \frac{\Lambda^{-1}}{\delta}(M_1(1+\delta)^{1/2} - M_2 - (\tau_{c}-\tau_{N})^{-1/2}\mathcal{O}(N^{-1/18})).
\end{equation}
The last term is strictly positive for $\delta$ large enough. For $0 < \tau_{N} = \tau_{c} - N^{-\beta}$ with $0<\beta<1/9$ and $N$ sufficiently large, we infer from \eqref {strict-non-vanishing} that there exists a positive constant $K$ such that
\[
D(w_{N},w_{N}) \geq K > 0.
\]
This contradicts \eqref{non-vanishing}. Hence vanishing does not occur.

\emph{Step 2: No dichotomy.} 
	Recall that the sequence $\{w_{N}\}$ is bounded uniformly in $L^{4/3}(\mathbb{R}^3)$, by Remark \ref{rem:coercivity}. In this step, we assume that $\{w_{N}\}$ is not relatively compact in $L^{1}(\mathbb{R}^3)$ up to translations. To obtain the desired contradiction, we make use of the dichotomy argument \cite{Lions-84a,Lions-84b}.
	
	\begin{lemma}[Strong Local Convergence]
		\label{lem:dichotomy}
		There exist a function $w^{(1)}\in L^{1}\cap L^{4/3}(\mathbb{R}^3)$ with $\int_{\mathbb{R}^3}w^{(1)}(x){\rm d}x=\nu\in (0,1)$, and sequences $\{R_{N}\}_{N\in\mathbb{N}}\subset\mathbb{R}_{+}$ with $R_{N}\to\infty$ and $\{y_{N}\}_{N\in\mathbb N}\subset\mathbb{R}^3$ such that, up to extraction of a subsequence,
		\begin{equation}\label{ineq:dichotomy}
		\lim_{N\to\infty}\int_{|x-y_{N}|\leq R_{N}}w_{N}(x){\rm d}x = \int_{\mathbb{R}^3}w^{(1)}(x){\rm d}x ,\quad \lim_{N\to\infty}\int_{R_{N}\leq |x-y_{N}| \leq 6R_{N}}w_{N}(x){\rm d}x = 0.
		\end{equation}
	\end{lemma}
	
	\begin{proof}
		We do not detail the proof of this lemma which uses concentration functions in the spirit of Lions \cite{Lions-84a,Lions-84b} as well as the strong local compactness of $w_{N}$. For instance, a similar argument has been detailed in \cite[Lemma 3.1]{Friesecke-03a} (see also \cite{Lewin-11,LenLew-10}).
	\end{proof}

	We remark that our model given by \eqref{HFB-functional} is invariant under translations. Thus, for the rest of the proof, we may assume that the sequence of translations in Lemma \ref{lem:dichotomy} is given by 
	$$
	y_{N} = 0 \text{ for all } N\geq 1.
	$$

	Let $0\leq \chi^{(1)} \leq 1$ be a fixed smooth function on $\mathbb{R}^3$ such that $\chi^{(1)}(x) \equiv 1$ for $|x|<1$ and $\chi^{(1)}(x) \equiv 0$ for $|x|\geq 2$. Given the sequence $\{R_{N}\}$ from Lemma \ref{lem:dichotomy}, we define the functions $\chi_{R_{N}}^{(1)}(x) = \chi^{(1)}(x/R_{N})$ and $\chi_{R_{N}}^{(2)}(x) = \sqrt{1-\chi_{R_{N}}^{(1)}(x)^{2}}$. Likewise, we define the sequences $\{\tilde{\gamma}_{N}^{(1)}\}_{N\in\mathbb N}$ and $\{\tilde{\gamma}_{N}^{(2)}\}_{N\in\mathbb N}$ by
	$$
	\tilde{\gamma}_{N}^{(i)}(x,y) = \chi_{N^{1/3}R_{N}}^{(i)}(x) \tilde{\gamma}_{N}(x,y) \chi_{N^{1/3}R_{N}}^{(i)}(y),\quad i\in\{1,2\},
	$$
	and $\rho_{\tilde{\gamma}_{N}^{(i)}}(x)=\tilde{\gamma}_{N}^{(i)}(x,x)$. We also set $w_{N}^{(i)}(x) = \chi_{R_{N}}^{(i)}(x)^{2}w_{N}(x)$. The direct term is separated as follows
	\begin{equation}\label{split:second-potential}
	D(w_{N},w_{N}) = D(w_{N}^{(1)},w_{N}^{(1)}) + D(w_{N}^{(2)},w_{N}^{(2)}) + 2D(w_{N}^{(1)},w_{N}^{(2)}).
	\end{equation}
	To show that the last term in \eqref{split:second-potential} is of order $1$, we write 
	$$
	\chi_{R_{N}}^{(2)}(y)^{2} = \chi_{3R_{N}}^{(2)}(y)^{2} + \chi_{R_{N}}^{(2)}(y)^{2} - \chi_{3R_{N}}^{(2)}(y)^{2}.
	$$
	Remark that $\chi_{R_{N}}^{(1)}(x)^{2}|x-y|^{-1}\chi_{3R_{N}}^{(2)}(y)^{2} \leq R_{N}^{-1}$ and
$$
\chi_{R_{N}}^{(1)}(x)^{2}[\chi_{R_{N}}^{(2)}(y)^{2}-\chi_{3R_{N}}^{(2)}(y)^{2}] \leq \mathbbm{1}(R_{N}\leq |y| \leq 6R_{N}).
$$
Thus, we use the Cauchy--Schwarz inequality (see \cite[Theorem 9.8]{LieLos-01}) and \eqref{ineq:HLS} to obtain
$$
D(w_{N}^{(1)},w_{N}^{(2)}) \leq R_{N}^{-1}\|w_{N}\|_{L^{1}}^{2} + C\|w_{N}\|_{L^{4/3}}^{4/3}\|w_{N}\|_{L^{1}}^{1/3}\|w_{N}\mathbbm{1}(R_{N}\leq |y| \leq 6R_{N})\|_{L^{1}}^{1/3}.
$$
The last term converges to $0$ as $N\to \infty$, thanks to Lemma \ref{lem:dichotomy} and the $L^{4/3}$-boundedness of $\{w_{N}\}$, by Remark \ref{rem:coercivity}.

Next, we split the kinetic energy. By using the IMS-type localization formula \cite{LieYau-88,FraLieSei-08} (see also \cite{LenLew-10}) and $\tr\tilde{\gamma}_{N} = \tr\gamma_{N} = N$, we find that
\begin{equation}\label{split:second-kinetic}
\ell_{N}^{-1}\frac{1}{N}\tr\sqrt{-\Delta+m^{2}}\gamma_{N} \geq \frac{1}{N}\tr\sqrt{-\Delta}\tilde{\gamma}_{N}  \geq \frac{1}{N}\tr\sqrt{-\Delta}\tilde{\gamma}_{N}^{(1)} + \frac{1}{N}\tr\sqrt{-\Delta}\tilde{\gamma}_{N}^{(2)} - \frac{C}{N^{1/3}R_{N}}.
\end{equation}
To deal with the second term on the right hand side in \eqref{split:second-kinetic}, we apply Lemma \ref{Chandrasekhar-lower-bound} with $g_{t}(x) = t^{3/4}\exp(-\pi t|x|^{2})$ and $\langle g_{t},\sqrt{-\Delta}g_{t} \rangle = 2t^{1/2}$. We obtain
\begin{align}\label{split:second-2}
\tr\sqrt{-\Delta}\tilde{\gamma}_{N}^{(2)} & \geq K_{\rm cl}\int_{\mathbb{R}^3}(\rho_{\tilde{\gamma}_{N}^{(2)}}\star g_{t}^{2})(x)^{4/3}{\rm d}x - 2t^{1/2}\int_{\mathbb{R}^3}\rho_{\tilde{\gamma}_{N}^{(2)}}(x){\rm d}x \nonumber\\
& \geq NK_{\rm cl}\int_{\mathbb{R}^3}(w_{N}^{(2)}\star g_{t_{N}}^{2})(x)^{4/3}{\rm d}x - 2t^{1/2}N,
\end{align}
where $t_{N} = tN^{2/3}$. Using \eqref{ineq:HLS} and noticing that $\|w_{N}^{(2)}\star g_{t_{N}}^{2}\|_{L^1} \leq \|w_{N}^{(2)}\|_{L^1} \leq \|w_{N}\|_{L^1} = 1$, we get
\begin{equation}\label{split:second-4}
K_{\rm cl}\int_{\mathbb{R}^3}(w_{N}^{(2)}\star g_{t_{N}}^{2})(x)^{4/3}{\rm d}x \geq \frac{\tau_c}{2} D(w_{N}^{(2)}\star g_{t_{N}}^{2},w_{N}^{(2)}\star g_{t_{N}}^{2}).
\end{equation}
On the other hand, we write
\begin{equation}\label{split:second-5}
D(w_{N}^{(2)},w_{N}^{(2)}) - D(w_{N}^{(2)}\star g_{t_{N}}^{2},w_{N}^{(2)}\star g_{t_{N}}^{2}) = \iint_{\mathbb{R}^3 \times \mathbb{R}^3} w_{N}^{(2)}(x)v_{N}(x-y)w_{N}^{(2)}(y){\rm d}x{\rm d}y
\end{equation}
with $v_{N}(x) = |x|^{-1} - (g_{t_{N}}^{2} \star |x|^{-1} \star g_{t_{N}}^{2})(x)$. Using Young's inequality, the integral in \eqref{split:second-5} can be bounded by $\|v_{N}\|_{L^{2}}\|w_{N}^{(2)}\|_{L^{4/3}}^{2}$. By a simple computation we have $\|v_{N}\|_{L^{2}} = Ct_{N}^{-1/4} = Ct^{-1/4}N^{-1/6}$. 
Combinning \eqref{split:second-2}, \eqref{split:second-4} and \eqref{split:second-5} we have
\begin{equation}\label{split:second-6}
\left(1-\frac{\kappa\pi}{4}\right)\frac{1}{N}\tr\sqrt{-\Delta}\tilde{\gamma}_{N}^{(2)} - \frac{\tau_{N}}{2}  D(w_{N}^{(2)},w_{N}^{(2)}) \geq R_{1},
\end{equation}
where we abbreviate by $R_{1}$ the error terms
\begin{equation}\label{split:second-7}
R_{1} := - \left(\frac{\kappa\pi}{4}+\epsilon\right)K_{\rm cl}\|w_{N}^{(2)}\star g_{t_{N}}^{2}\|_{L^{4/3}}^{4/3} + \epsilon \frac{1}{N}\tr\sqrt{-\Delta}\tilde{\gamma}_{N}^{(2)} -  Ct^{-1/4}N^{-1/6}\|w_{N}^{(2)}\|_{L^{4/3}}^{2} - 2t^{1/2}.
\end{equation}
By Daubechies' inequality \cite{Daubechies-83} we have 
\begin{equation}\label{split:Dau}
\tr\sqrt{-\Delta}\tilde{\gamma}_{N}^{(2)} \geq 1.6q^{-1/3}\int_{\mathbb{R}^3}\rho_{\tilde{\gamma}_{N}^{(2)}}(x)^{4/3}{\rm d}x = 1.6q^{-1/3}N\int_{\mathbb{R}^3}w_{N}^{(2)}(x)^{4/3}{\rm d}x.
\end{equation}
Hence, optimizing the last two terms in \eqref{split:second-7} with respect to $t$ and choosing $\epsilon=CN^{-1/9}$ for a suitable constant $C$, we get
\begin{equation}\label{split:second-8}
R_{1} \geq - \left(\frac{\kappa\pi}{4}+\epsilon\right)K_{\rm cl}\|w_{N}^{(2)}\star g_{t_{N}}^{2}\|_{L^{4/3}}^{4/3} \geq - \left(\frac{\kappa\pi}{4}+\epsilon\right)K_{\rm cl}\|w_{N}^{(2)}\|_{L^{4/3}}^{4/3} = o(1)_{N\to\infty}.
\end{equation}
Here we have used the fact that $w_{N}$ (and hence $w_{N}^{(2)}$) is bounded uniformly in $L^{4/3}(\mathbb{R}^3)$ and that $\kappa=\tau_{N}N^{-2/3}=\mathcal{O}(N^{-2/3})$. In summary, from \eqref{split:second-potential}--\eqref{split:second-kinetic}, \eqref{split:second-6}--\eqref{split:second-8} and \eqref{energy:lower} we have derived the following estimate
\begin{align}
\ell_{N}^{-1}\frac{E_{\tau_{N}}^{\rm HFB}(N)}{N} & \geq \left(1-\frac{\kappa\pi}{4}\right)\ell_{N}^{-1}\frac{1}{N}\tr\sqrt{-\Delta+m^{2}}\gamma_{N} - \frac{\tau_{N}}{2}  D(w_{N},w_{N}) \nonumber \\
& \geq 
\left(1-\frac{\kappa\pi}{4}\right)\frac{1}{N}\tr\sqrt{-\Delta}\tilde{\gamma}_{N}^{(1)} - \frac{\tau_{N}}{2}  D(w_{N}^{(1)},w_{N}^{(1)}) + o(1)_{N\to\infty}. \label{split:critical}
\end{align}

It follows from Lemma \ref{lem:dichotomy} that $w_{N}^{(1)}$ converges to $w^{(1)}$ weakly in $L^{4/3}(\mathbb{R}^3)$ and strongly in $L^{1}(\mathbb{R}^3)$. In fact, $w_{N}^{(1)}$ converges to $w^{(1)}$ strongly in $L^{r}(\mathbb{R}^3)$ for $1\leq r <4/3$ because of $L^{4/3}(\mathbb{R}^3)$-boundedness of $\{w_{N}\}$, by Remark \ref{rem:coercivity}. Thus, by the Hardy--Littlewood--Sobolev inequality (see \cite[Theorem 4.3]{LieLos-01}) we have
\begin{equation}\label{split:second-0}
\lim_{N\to\infty}D(w_{N}^{(1)},w_{N}^{(1)}) = D(w^{(1)},w^{(1)}).
\end{equation}
On the other hand, we note that the inequality \eqref{ineq:gamma-alpha} implies the Pauli exclusion principle \cite[Theorem 3.2]{LieSei-10}
$$
0\leq \gamma_{N} \leq 1.
$$
This property is invariant under scaling as well as under restricting on a domain. Hence we may apply Lemma \ref{lem:LT} to the sequence $\tilde{\gamma}_{N}^{(1)}(x,y) = \ell_{N}^{-3}\gamma_{N}^{(1)}(\ell_{N}^{-1}x,\ell_{N}^{-1}y)$ together with the weak convergence $w_{N}^{(1)}\rightharpoonup w^{(1)}$ in $L^{4/3}(\mathbb{R}^3)$. We obtain
\begin{equation}\label{split:second-1}
\lim_{N\to\infty}\frac{1}{N}\tr\sqrt{-\Delta}\tilde{\gamma}_{N}^{(1)} \geq K_{\rm cl}\int_{\mathbb{R}^3}w^{(1)}(x)^{4/3}{\rm d}x.
\end{equation}
Taking the limit $N\to\infty$ in \eqref{split:critical} and using \eqref{split:second-0}, \eqref{split:second-1} together with the asymptotic formula for $E_{\tau_{N}}^{\rm HFB}(N)$ in Lemma \ref{lem:HFB-energy} we obtain
\begin{equation}\label{non-dichotomy}
0 = \lim_{N\to\infty}\ell_{N}^{-1}\frac{E_{\tau_{N}}^{\rm HFB}(N)}{N} \geq \mathcal{E}_{\tau_c}^{\rm Ch}(w^{(1)})|_{m=0} \geq E_{\tau_{c}}^{\rm Ch}(\nu)|_{m=0} = \nu E_{\tau_{c}\nu^{2/3}}^{\rm Ch}(1)|_{m=0} = 0.
\end{equation}
It follows from \eqref{non-dichotomy} that $w^{(1)}$ is a minimizer for $E_{\tau_{c}}^{\rm Ch}(\nu)|_{m=0}$. But this contradicts the fact that the variational problem $E_{\tau_{c}}^{\rm Ch}(\nu)|_{m=0}=0$ has no minimizer for any $0<\nu<1$, which is due to the positivity of the direct term (see \cite[Theorem 9.8]{LieLos-01}). Hence, dichotomy does not occur.

\emph{Step 3: Conclusion.} We conclude that, up to translations, the sequence $\{w_{N}\}$ is relatively compact in $L^{1}(\mathbb{R}^3)$. Hence, there exist a subsequence of $\{w_{N}\}$, still denoted by $\{w_{N}\}$, and a function $w\in L^{1}\cap L^{4/3}(\mathbb{R}^3)$ with $\int_{\mathbb{R}^3}w(x){\rm d}x=1$ such that $w_{N}$ converges to $w$ strongly in $L^1(\mathbb{R}^3)$, weakly in $L^{4/3}(\mathbb{R}^3)$ and pointwise almost everywhere in $\mathbb{R}^3$. In fact, $w_{N}$ converges to $w$ strongly in $L^{r}(\mathbb{R}^{3})$ for $1\leq r<4/3$ because of $L^{4/3}(\mathbb{R}^3)$-boundedness of $\{w_{N}\}$, by Remark \ref{rem:coercivity}. Applying Lemma \ref{lem:LT} to the sequence $\tilde{\gamma}_{N}(x,y) = \ell_{N}^{-3}\gamma_{N}(\ell_{N}^{-1}x,\ell_{N}^{-1}y)$ and using the Hardy--Littlewood--Sobolev inequality (see \cite[Theorem 4.3]{LieLos-01}), we obtain
\begin{equation}\label{fake-Q}
	0 = \lim_{N\to\infty}\ell_{N}^{-1}\frac{E_{\tau_{N}}^{\rm HFB}(N)}{N} \geq \mathcal{E}_{\tau_c}^{\rm Ch}(w)|_{m=0} \geq E_{\tau_{c}}^{\rm Ch}(1)|_{m=0} = 0.
\end{equation}
Here we have used the asymptotic formula for $E_{\tau_{N}}^{\rm HFB}(N)$ in Lemma \ref{lem:HFB-energy} and \eqref{ineq:HLS} for the last estimate. It follows from \eqref{fake-Q} that $w$ is a minimizer for $E_{\tau_{c}}^{\rm Ch}(1)|_{m=0}=0$. In other words, $w$ is an optimizer for \eqref{ineq:HLS} with $\int_{\mathbb{R}^3}w(x){\rm d}x=1$. We recall that \eqref{ineq:HLS} admits a unique (up to translations and dilations) normalized optimizer which satisfies \eqref{eq:massless-neutron-star} (after a suitable scaling). Therefore, we have
$$
w(x)=b^{3}Q(bx)
$$
for some $b>0$, and for $Q\in L^{1}\cap L^{4/3}(\mathbb{R}^{3})$ the unique non-negative radially symmetric decreasing solution to the equation \eqref{eq:massless-neutron-star}. Note that $\int_{\mathbb{R}^3}Q(x){\rm d}x = \int_{\mathbb{R}^3}w(x){\rm d}x=1$. Hence, we deduce from \eqref{eq:massless-neutron-star} and \eqref{fake-Q} that $Q$ satisfies \eqref{cond:LE}.

We shall show that $b=1$ and hence $w\equiv Q$. We first apply Lemma \ref{Chandrasekhar-lower-bound} with $g_{t}(x) = t^{3/4}\exp(-\pi t|x|^{2})$ and $\langle g_{t},\sqrt{-\Delta}g_{t} \rangle = 2t^{1/2}$ to obtain
\begin{equation}\label{conv:critical-1}
\tr\sqrt{-\Delta+m^{2}}\gamma_{N} \geq \int_{\mathbb{R}^3}j_{m}((\rho_{\gamma_{N}}\star g_{t}^{2})(x)){\rm d}x - 2Nt^{1/2}.
\end{equation}
By a simple scaling $\rho_{\gamma_{N}}(x) = \ell_{N}^{3}w_{N}(N^{-1/3}\ell_{N} x)$ using \eqref{kinetic} we have
\begin{equation}\label{conv:critical-2}
\int_{\mathbb{R}^3}j_{m}((\rho_{\gamma_{N}}\star g_{t}^{2})(x)){\rm d}x = N\ell_{N} \int_{\mathbb{R}^3}j_{m\ell_{N}^{-1}}((w_{N}\star g_{t_{N}}^{2})(x)){\rm d}x,
\end{equation}
where $t_{N} = tN^{2/3}\ell_{N}^{-2}$. Now we define the function $\tilde{j}_{m}$ by
\begin{align*}
\tilde{j}_{m}(\rho) : & = \frac{q}{(2\pi)^3}\int_{|p|<(6\pi^{2}\rho/q)^{1/3}}\frac{1}{\sqrt{|p|^{2}+m^{2}}}{\rm d}p \\
& = \frac{q}{4\pi^{2}}\left[\eta\sqrt{\eta^{2}+m^{2}}-m^{2}\ln\left(\frac{\eta+\sqrt{\eta^{2}+m^{2}}}{m}\right)\right],\quad \eta=\left(\frac{6\pi^{2}\rho}{q}\right)^{1/3}.
\end{align*}
Then we have 
\begin{equation}\label{conv:critical-3}
\int_{\mathbb{R}^3}j_{m\ell_{N}^{-1}}((w_{N}\star g_{t_{N}}^{2})(x)){\rm d}x \geq K_{\rm cl}\int_{\mathbb{R}^3}(w_{N}\star g_{t_{N}}^{2})(x)^{4/3}{\rm d}x + \frac{m^{2}\ell_{N}^{-2}}{2}\int_{\mathbb{R}^3}\tilde{j}_{m\ell_{N}^{-1}}((w_{N}\star g_{t_{N}}^{2})(x)){\rm d}x,
\end{equation}
which follows from the operator inequality 
\begin{equation}\label{ineq:operator}
\sqrt{|p|^{2}+m^{2}}\geq |p|+\frac{m^{2}}{2\sqrt{|p|^{2}+m^{2}}}.
\end{equation}
Using \eqref{ineq:HLS} and noticing that $\|w_{N}\star g_{t_{N}}^{2}\|_{L^1} \leq \|w_{N}\|_{L^1} = 1$ we get
\begin{equation}\label{conv:critical-4}
K_{\rm cl}\int_{\mathbb{R}^3}(w_{N}\star g_{t_{N}}^{2})(x)^{4/3}{\rm d}x \geq \frac{\tau_c}{2} D(w_{N}\star g_{t_{N}}^{2},w_{N}\star g_{t_{N}}^{2}).
\end{equation}
On the other hand, we write
\begin{equation}\label{conv:critical-5}
D(w_{N},w_{N}) - D(w_{N}\star g_{t_{N}}^{2},w_{N}\star g_{t_{N}}^{2}) = \iint_{\mathbb{R}^3 \times \mathbb{R}^3} w_{N}(x)v_{N}(x-y)w_{N}(y){\rm d}x{\rm d}y
\end{equation}
with $v_{N}(x) = |x|^{-1} - (g_{t_{N}}^{2} \star |x|^{-1} \star g_{t_{N}}^{2})(x)$. Using Young's inequality, the integral in \eqref{conv:critical-5} can be bounded by $\|v_{N}\|_{L^{2}}\|w_{N}\|_{L^{4/3}}^{2}$. By a simple computation we have $\|v_{N}\|_{L^{2}} = Ct_{N}^{-1/4} = Ct^{-1/4}N^{-1/6}\ell_{N}^{1/2}$. 
Combining \eqref{conv:critical-1}--\eqref{conv:critical-5} together with \eqref{energy:lower} we have
\begin{align}
\frac{E_{\tau_{N}}^{\rm HFB}(N)}{N} & \geq \left(1-\frac{\kappa\pi}{4}\right)\frac{1}{N}\tr\sqrt{-\Delta+m^{2}}\gamma_{N} - \frac{\kappa}{2}  D(\rho_{\gamma_{N}},\rho_{\gamma_{N}}) \nonumber \\
& \geq \left(1-\frac{\kappa\pi}{4}-\epsilon\right)\ell_{N}^{-1}\frac{m^{2}}{2}\int_{\mathbb{R}^3}\tilde{j}_{m\ell_{N}^{-1}}((w_{N}\star g_{t_{N}}^{2})(x)){\rm d}x + \ell_{N}\frac{\tau_c-\tau_{N}}{2}D(w_{N},w_{N}) + R_{2} \label{conv:critical-6}
\end{align}
where we abbreviate by $R_{2}$ the remainder terms 
\begin{align}\label{reminders}
R_{2} & := - \left(\frac{\kappa\pi}{4}+\epsilon\right)\ell_{N}K_{\rm cl}\|w_{N}\star g_{t_{N}}^{2}\|_{L^{4/3}}^{4/3} + \frac{\epsilon}{N}\tr\sqrt{-\Delta+m^{2}}\gamma_{N} \nonumber \\
& \qquad - Ct^{-1/4}N^{-1/6}\ell_{N}^{3/2}\|w_{N}\star g_{t_{N}}^{2}\|_{L^{4/3}}^{2} - 2t^{1/2}.
\end{align}
By Daubechies' inequality \cite{Daubechies-83} we have 
$$
\tr\sqrt{-\Delta}\gamma_{N} \geq 1.6q^{-1/3}\int_{\mathbb{R}^3}\rho_{\gamma_{N}}(x)^{4/3}{\rm d}x = 1.6q^{-1/3}N\ell_{N}\int_{\mathbb{R}^3}w_{N}(x)^{4/3}{\rm d}x.
$$
Optimizing the last two terms in \eqref{reminders} with respect to $t$, whence $t_{N}=CN^{4/9}\to \infty$ as $N\to\infty$, and choosing $\epsilon=Cq^{1/3}N^{-1/9}$ for a suitable constant $C$, we get
\begin{equation}\label{conv:critical-7}
R_{2} \geq - \left(\frac{\kappa\pi}{4}+\epsilon\right)\ell_{N}K_{\rm cl}\|w_{N}\star g_{t_{N}}^{2}\|_{L^{4/3}}^{4/3} \geq -\ell_{N}^{-1} o(1)_{N\to\infty}.
\end{equation}
Here we have used the fact that $w_{N}$ is bounded uniformly in $L^{4/3}(\mathbb{R}^3)$ and that $\kappa=\tau_{N}N^{-2/3}=\mathcal{O}(N^{-2/3})$. Putting \eqref{conv:critical-6} and \eqref{conv:critical-7} together we obtain
\begin{equation}\label{liminf:energy}
\frac{\ell_{N}}{\Lambda}\cdot\frac{E_{\tau_{N}}^{\rm HFB}(N)}{N} \geq (1+o(1)_{N\to\infty})\frac{m^{2}}{2\Lambda}\int_{\mathbb{R}^3}\tilde{j}_{m\ell_{N}^{-1}}((w_{N}\star g_{t_{N}}^{2})(x)){\rm d}x + \frac{\Lambda}{2} D(w_{N},w_{N}) + o(1)_{N\to\infty}.
\end{equation}
Now we note that the strong convergence $w_{N}\to w$ in $L^{r}(\mathbb{R}^{3})$ for $1\leq r<4/3$ implies the strong convergence $w_{N}\star g_{t_{N}}^{2}\to w$ in $L^{r}(\mathbb{R}^{3})$ for $1\leq r<4/3$. This follows from the fact that $w\star g_{t_{N}}^{2} \to w$ strongly in $L^{r}(\mathbb{R}^{3})$ for $1\leq r<4/3$ (recall that $t_{N}\to\infty$) and that
\begin{align*}
\|w_{N}\star g_{t_{N}}^{2} - w\|_{L^r} & \leq \|(w_{N}-w)\star g_{t_{N}}^{2}\|_{L^r} + \|w\star g_{t_{N}}^{2} - w\|_{L^r} \\
& \leq \|(w_{N}-w)\|_{L^r} + \|w\star g_{t_{N}}^{2} - w\|_{L^r}.
\end{align*}
Here we have used Minkowski's inequality and Young's inequality. Thus, we conclude that $w_{N}\star g_{t_{N}}^{2}\to w$ pointwise almost everywhere in $\mathbb{R}^3$, up to extraction of a subsequence. By Fatou's lemma we have 
\begin{align}\label{liminf kinetic}
\liminf_{N\to\infty}\int_{\mathbb{R}^3}\tilde{j}_{m\ell_{N}^{-1}}((w_{N}\star g_{t_{N}}^{2})(x)){\rm d}x \geq \frac{q}{4\pi^{2}}\int_{\mathbb{R}^3}\theta(x)^{2}{\rm d}x = \frac{9}{8bK_{\rm cl}}\int_{\mathbb{R}^3}Q(x)^{2/3}{\rm d}x,
\end{align}
where $\theta=(6\pi^{2} w/q)^{1/3}$. On the other hand, by the Hardy--Littlewood--Sobolev inequality (see \cite[Theorem 4.3]{LieLos-01}) we have
\begin{align}\label{liminf:D}
\lim_{N\to\infty}D(w_{N},w_{N}) = D(w,w) = bD(Q,Q) = 2b.
\end{align}
Thus, after passing to the limit $N\to\infty$ in \eqref{liminf:energy} and using the asymptotic formula for $E_{\tau_{N}}^{\rm HFB}(N)$ in Lemma \ref{lem:HFB-energy} we obtain
\begin{equation}\label{eq:b}
2\Lambda \geq \frac{9m^{2}}{16b\Lambda K_{\rm cl}}\int_{\mathbb{R}^3}Q(x)^{2/3}{\rm d}x + b\Lambda .
\end{equation}
It is elementary to check that
$$
\inf_{\lambda>0}\left(\frac{9m^{2}}{16\lambda K_{\rm cl}}\int_{\mathbb{R}^3}Q(x)^{2/3}{\rm d}x + \lambda\right) = 2\Lambda
$$
with the unique optimal value $\lambda = \Lambda$. Therefore, the equality in \eqref{eq:b} must occur and hence $b = 1$. We thus have shown that, up to extraction of a subsequence, $w_{N}$ converges to the unique Lane--Emden solution satisfying \eqref{cond:LE}--\eqref{eq:massless-neutron-star}. This completes the proof of Theorem \ref{thm:behavior}.

\section*{Acknowledgements}

The manuscript was completed when the author was visiting the Mittag--Leffler Institute for the semester program \emph{Spectral Methods in Mathematical Physics}. The author would like to thank the organizers for their warm hospitality. He also thanks E. Lenzmann for helpful comments. The research received funding from the Deutsche Forschungsgemeinschaft (DFG, German Research Foundation) under Germany's Excellence Strategy -- EXC-2111 -- 390814868.


\medskip

\begin{thebibliography}{10}
	
	\bibitem{AscFroGraSchTro-02} W.~Aschbacher, J.~Fr{\"o}hlich, G.~Graf, K.~Schnee, and M.~Troyer, {\em{Symmetry breaking regime in the nonlinear Hartree equation}}, Journal of
	Mathematical Physics, 43 (2002), pp.~3879--3891.
	
	\bibitem{BacFroJon-09} V.~Bach, J.~Fr{\"o}hlich, and B.~L.~G. Jonsson, {\em {Bogolubov--{H}artree--{F}ock mean field theory for neutron stars and other systems with attractive interactions}}, Journal of Mathematical Physics, 50 (2009), p.~102102.
	
	\bibitem{BacLieSol-94} V.~Bach, E.~H. Lieb, and J.~P. Solovej, {\em {Generalized {H}artree--{F}ock theory and the {H}ubbard model}}, Journal of Statistical Physics, 76 (1994), pp.~3--89.
	
	\bibitem{Chandrasekhar-31a} S.~Chandrasekhar, {\em {The maximum mass of ideal white dwarfs}}, Astrophysical Journal, 74 (1931), pp.~81--82.
	
	\bibitem{Daubechies-83} I.~Daubechies, {\em {An uncertainty principle for fermions with generalized kinetic energy}}, Communications in Mathematical Physics, 90 (1983), pp.~511--520.
	
	\bibitem{FouLewSol-18} S.~Fournais, M.~Lewin, and J.~P. Solovej, {\em The semi-classical limit of large fermionic systems}, Calculus of Variations and Partial Differential Equations, 57 (2018), p.~105.
	
	\bibitem{FraLen-13} R.~L. Frank and E.~Lenzmann, {\em {Uniqueness of non-linear ground states for fractional Laplacians in $\mathbb {R}$}}, Acta mathematica, 210 (2013), pp.~261--318.
	
	\bibitem{FraLenSil-16} R.~L. Frank, E.~Lenzmann, and L.~Silvestre, {\em {Uniqueness of radial solutions for the fractional Laplacian}}, Communications on Pure and Applied Mathematics, 69 (2016), pp.~1671--1726.
	
	\bibitem{FraLieSei-08} R.~L. Frank, E.~H. Lieb, and R.~Seiringer, {\em {Hardy--Lieb--Thirring} inequalities for fractional {S}chr{\"o}dinger operators}, Journal of the American Mathematical Society, 21 (2008), pp.~925--950.
	
	\bibitem{Friesecke-03a} G.~Friesecke, {\em {The multiconfiguration equations for atoms and molecules: charge quantization and existence of solutions}}, Archive for Rational Mechanics and Analysis, 169 (2003), pp.~35--71.
	
	\bibitem{FroLen-07f} J.~Fr{\"o}hlich and E.~Lenzmann, {\em {Dynamical collapse of white dwarfs in {H}artree- and {H}artree--{F}ock theory}}, Communications in Mathematical Physics, 274 (2007), pp.~737--750.
	
	\bibitem{GadBarHan-80} S.~R. Gadre, L.~J. Bartolotti, and N.~C. Handy, {\em Bounds for {C}oulomb energies}, The Journal of Chemical Physics, 72 (1980), pp.~1034--1038.
	
	\bibitem{GuoZen-17} Y.~Guo and X.~Zeng, {\em {Ground States of Pseudo-Relativistic Boson Stars under the Critical Stellar Mass}}, Annales de l'Institut Henri Poincar\'e (C) Analyse Non Lin{\'e}aire, 34 (2017), pp.~1611--1632.
	
	\bibitem{Hainzl-10} C.~Hainzl, {\em On the static and dynamical collapse of white dwarfs}, Entropy and the Quantum: Arizona School of Analysis with Applications, March 16-20, 2009, University of Arizona, 529 (2010), pp.~189--202.
	
	\bibitem{HaiLenLewSch-10} C.~Hainzl, E.~Lenzmann, M.~Lewin, and B.~Schlein, {\em {On Blowup for Time-Dependent Generalized {H}artree--{F}ock equations}}, Annales Henri Poincar{\'e}, 11 (2010), pp.~1023--1052.
	
	\bibitem{HaiSch-09} C.~Hainzl and B.~Schlein, {\em {Stellar Collapse in the Time Dependent {{H}artree{--}{F}ock} approximation}}, Communications in Mathematical Physics, 287 (2009), pp.~705--717.
	
	\bibitem{Herbst-77} I.~W. Herbst, {\em {Spectral theory of the operator {$(p\sp{2}+m\sp{2})\sp{1/2}-Ze\sp{2}/r$}}}, Communications in Mathematical Physics, 53 (1977), pp.~285--294.
	
	\bibitem{Kato-95} T.~Kato, {\em {Perturbation theory for linear operators}}, Springer, second~ed., 1995.
	
	\bibitem{Lane-70} J.~H. Lane, {\em On the theoretical temperature of the sun; under the hypothesis of a gaseous mass maintaining its volume by its internal heat, and depending on the laws of gases as known to terrestrial experiment}, American Journal of Science and Arts, 50 (1870), pp.~57--74.
	
	\bibitem{LenLew-10} E.~Lenzmann and M.~Lewin, {\em {Minimizers for the {H}artree--{F}ock--{B}ogoliubov Theory of Neutron Stars and White Dwarfs}}, Duke Mathematical Journal, 152 (2010), pp.~257--315.
	
	\bibitem{LenLew-11} \leavevmode\vrule height 2pt depth -1.6pt width 23pt, {\em {On singularity formation for the ${L}^2$-critical {B}oson star equation}}, Nonlinearity, 24
	(2011), pp.~3515--3540.
	
	\bibitem{Lewin-11} M.~Lewin, {\em {Geometric methods for nonlinear many-body quantum systems}}, Journal of Functional Analysis, 260 (2011), pp.~3535--3595.
	
	\bibitem{Lieb-83} E.~H. Lieb, {\em {On the lowest eigenvalue of the {L}aplacian for the intersection of two domains}}, Inventiones Mathematicae, 74 (1983), pp.~441--448.
	
	\bibitem{LieLos-01} E.~H. Lieb and M.~Loss, {\em {Analysis}}, vol.~14 of {Graduate Studies in Mathematics}, American Mathematical Society, Providence, RI, 2nd~ed., 2001.
	
	\bibitem{LieOxf-80} E.~H. Lieb and S.~Oxford, {\em {Improved lower bound on the indirect {C}oulomb energy}}, International Journal of Quantum Chemistry, 19 (1980), pp.~427--439.
	
	\bibitem{LieSei-10} E.~H. Lieb and R.~Seiringer, {\em {The {S}tability of {M}atter in {Q}uantum {M}echanics}}, Cambridge University Press, 2010.
	
	\bibitem{LieThi-84} E.~H. Lieb and W.~E. Thirring, {\em {Gravitational collapse in quantum mechanics with relativistic kinetic energy}}, Annals of Physics, 155 (1984), pp.~494--512.
	
	\bibitem{LieYau-87} E.~H. Lieb and H.-T. Yau, {\em {{The {C}handrasekhar theory of Stellar Collapse as the Limit of Quantum Mechanics}}}, Communications in Mathematical Physics, 112 (1987), pp.~147--174.
	
	\bibitem{LieYau-88} \leavevmode\vrule height 2pt depth -1.6pt width 23pt, {\em {The stability and instability of relativistic matter}}, Communications in Mathematical Physics, 118 (1988), pp.~177--213.
	
	\bibitem{Lions-84a} P.-L. Lions, {\em {The concentration-compactness principle in the calculus of variations. {T}he locally compact case, {P}art {I}}}, Annales de l'Institut Henri Poincar\'e (C) Analyse Non Lin\'eaire, 1 (1984), pp.~109--149.
	
	\bibitem{Lions-84b} \leavevmode\vrule height 2pt depth -1.6pt width 23pt, {\em {The concentration-compactness principle in the calculus of variations. {T}he locally compact case, {P}art {II}}}, Annales de l'Institut Henri Poincar\'e (C) Analyse Non Lin\'eaire, 1 (1984), pp.~223--283.
	
	\bibitem{Maeda-10} M.~Maeda, {\em {On the symmetry of the ground states of nonlinear {S}chr{\"o}dinger equations with potential}}, Advanced Nonlinear Studies, 10 (2010), pp.~895--925.
	
	\bibitem{Nguyen-17a} D.-T. Nguyen, {\em {Blow-up profile of ground states for the critical boson star}}, arXiv preprint arXiv:1703.10324v1,  (2017).
	
	\bibitem{Nguyen-17b} \leavevmode\vrule height 2pt depth -1.6pt width 23pt, {\em {On Blow-up Profile of Ground States of Boson Stars with External Potential}}, Journal of Statistical Physics, 169 (2017), pp.~395--422.
	
	\bibitem{Nguyen-19f} \leavevmode\vrule height 2pt depth -1.6pt width 23pt, {\em {Blow-up Profile of Neutron Stars in the Chandrasekhar theory}}, Journal of Mathematical Physics, 60 (2019), p.~071508.
	
	\bibitem{Nguyen-19b} \leavevmode\vrule height 2pt depth -1.6pt width 23pt, {\em {Many-Body Blow-Up Profile of Boson Stars with External Potentials}}, Review in Mathematical Physics, 31 (2019), p.~1950034.
	
	\bibitem{Straumann-12} N.~Straumann, {\em General relativity and relativistic astrophysics}, Springer Science \& Business Media, 2012.
	
	\bibitem{Thomas-27} L.~H. Thomas, {\em {The calculation of atomic fields}}, Mathematical Proceedings of the Cambridge Philosophical Society, 23 (1927), pp.~542--548.
	
	\bibitem{Weinberg-72} S.~Weinberg, {\em Gravitation and cosmology: principles and applications of the general theory of relativity}, vol.~1, Wiley New York, 1972.
	
	\bibitem{YanYan-17} J.~Yang and J.~Yang, {\em {Existence and mass concentration of pseudo-relativistic {H}artree equation}}, Journal of Mathematical Physics, 58 (2017), p.~081501.
	
\end{thebibliography}
\end{document}